\documentclass[a4paper]{amsart}

\usepackage{float}
\usepackage{amssymb, amsfonts, amsmath, amscd, mathrsfs, mathtools, bbm,comment,subfigure,color}
\usepackage[draft]{changes}
\usepackage[outerbars]{changebar}

\newtheorem{proposition}{Proposition}

\newtheorem{corollary}{Corollary}
\newtheorem{lemma}{Lemma}

\newcommand{\N}{\ensuremath{\mathbb{N}}}
\newcommand{\Z}{\ensuremath{\mathbb{Z}}}

\newcommand{\R}{\ensuremath{\mathbb{R}}}
\newcommand{\C}{\ensuremath{\mathbb{C}}}
\newcommand{\E}{\ensuremath{\mathbb{E}}}
\renewcommand{\P}{\ensuremath{\mathbb{P}}}

\newcommand{\ind}[1]	{\ensuremath{\mathbbm{1}_{\left\{ {#1} \right\}}}}
\newcommand{\abs}[1]	{\ensuremath{\left| {#1} \right|}}

\newcommand{\atan}[1]	{\ensuremath{\mathop{\mathrm{ArcTan}} \left( {#1} \right)}}
\newcommand{\var}[1]	{\ensuremath{\mathop{\mathrm{var}}_{\; #1}}}

\newcommand{\I}		{\ensuremath{\mathbb{I}}}
\newcommand{\diff}		{\ensuremath{\mathop{}\mathopen{}\mathrm{d}}}

\date{\today}

\begin{document}
\title[Sojourn in an $M/M/1$ PS queue]{On the sojourn of an arbitrary customer in an $M/M/1$ Processor Sharing Queue}

\author{Fabrice Guillemin and Veronica Quintuna}
\address{Orange Labs2, CNC/NCA, Avenue Pierre Marzin, 22300 Lannion, France}


\begin{abstract}
In this paper, we consider the number of both arrivals and departures seen by  a tagged customer while in service in  a classical $M/M/1$ processor sharing queue. By exploiting the underlying orthogonal structure of this queuing system revealed in an earlier study, we  compute the distributions of these two quantities and prove that they are equal in distribution. We moreover derive the asymptotic behavior of this common distribution. The knowledge of the number of departures seen by a tagged customer allows us to test the validity of an approximation, which consists of assuming that the tagged customer is randomly served among those customers in the residual busy period of the queue following the arrival of the tagged customer. A numerical evidence shows that this approximation is reasonable for moderate values of the number of departures, given that the asymptotic behaviors of the distributions are very different even if the exponential decay rates are equal.
\end{abstract}

\maketitle

\section{Introduction}

The processor sharing (PS) policy  is a well known service discipline, which was introduced in the 1960's by Kleinrock~\cite{Klein1} in the performance evaluation of computer networks (see also~\cite{Klein0}). With this service discipline, jobs in the queue are served in an egalitarian way. Thus, if there are $n$ jobs in the queue, each job receives the $1/n$ fraction  of the server capacity. The $M/G/1$-PS queue has been studied in the queuing literature by several authors, see for instance~\cite{Ott,Schassberger,Yashkov}, who computed the Laplace transform of the sojourn time of a tagged customer conditioned on the service time. For the specific case of the $M/M/1$ queue, it is possible to obtain an explicit expression of the sojourn time distribution of an arbitrary customer~\cite{Coffman,Morrison}. In~\cite{FGBJ}, an orthogonal structure involving Pollaczek polynomials was introduced to solve the infinite differential system satisfied by the vector composed of the distributions of the sojourn time of a customer conditioned on the number of customers in the system upon arrival. The $M/G/1$-PS system has then been further extended to account of permanent customers~\cite{Boxma,FGFS,Yashkov2}.

From a practical point view, the processor sharing discipline has gained renewed interest in the study of resource sharing in the Internet. As a matter of fact, the PS discipline can be used to model how TCP connections share the bandwidth of a bottleneck in a packet network~\cite{JRLM}. More recently, in the context of cloud platforms, the processor sharing can reflect how the capacity of a multi-core platform is shared among several tenants~\cite{veronica2}. 

In this paper, instead of directly computing the distribution of the sojourn time of a customer, we study the number of both arrivals and departures seen by a tagged customer while it is in service, denoted by $\alpha$ and $\delta$, respectively. We explicitly compute the distributions of these two quantities in the stationary regime. We notably establish that $\alpha \stackrel{d}{=} \delta$ (equality in distribution). As a byproduct, we can recover the distribution of the sojourn time of the tagged customer in the queue. To compute the distribution of $\delta$, we use the orthogonal structure underlying the $M/M/1$-PS queue, namely Pollaczek polynomials and their associated orthogonality measure~\cite{FGBJ}.

The knowledge of the distribution of $\delta$ allows us to easily test the accuracy of the following approximation: Since at each departure, all customers have equal chance of leaving the system, because of the memory-less property of the exponential distribution, we can suppose that when a tagged customer enters the queue, this customer is randomly served among those customers served in the residual busy period (i.e., the busy period starting at the arrival epoch of the tagged customer and ending  when the queue empties). Under this assumption, we compute the distribution of the number $\widetilde{b}$ of customers leaving the system before the tagged customer is randomly picked. Even if this approximation seems to be rough at first glance, numerical results show that $\widetilde{b}$ is a reasonable upper bound for $\delta$. The motivation for this approximation is to develop a method of approximating the sojourn time of a batch in an $M^{[X]}/M/1$-PS queue. While the Laplace transform of the sojourn time of a job has been computed in~\cite{Veronica1} when batches are geometrically distributed, results for the sojourn time of an entire batch is much more challenging.

The organization of this paper is as follows: In Section~\ref{model}, we introduce the model and the various random variables. In particular, these random variables are related to an absorbed Markov discrete time, whose transition matrix introduces a selfadjoint operator in an ad-hoc Hilbert space. The spectral properties of this operator are studied in Section~\ref{spectral}  and the distributions of the random variables are computed in Section~\ref{distributions}.  Numerical results are presented in Section~\ref{numerical}. Further research directions are discussed in Section~\ref{conclusion}.

\section{Model description}
\label{model}

We consider a classical $M/M/1$ Processor sharing queue with arrival rate $\rho$ and unit service rate. We assume that a tagged customer arrives at the queue at time $t=0$ and that there are $n$ customers in the queue upon arrival of this tagged customer. We introduce the discrete-time process $(N_k)$ describing the number of customers in the queue other than the tagged one at the departures or arrivals of customers.  When the tagged customer completes its service, the process $N_k$ is absorbed in some state, denoted by $-1$. The index $k$ is thus the number of departures from the queue or arrivals at this queue before the process $(N_k)$ gets absorbed. We set $N_0=n$ since we assume that there are $n$ customers in the queue upon arrival of the tagged customer.

The state space of the process $(N_k)$ is $\{-1,0, 1,2, \ldots\}$ and $(N_k)$ is a discrete-time Markov chain with transition matrix $\mathcal{A}$ given by
$$
\mathcal{A} = \begin{pmatrix}
1 & 0 & 0 & 0 & 0 & \ldots \\
\frac{1}{1+\rho} & 0 & \frac{\rho}{1+\rho} & 0 & 0 & \ldots \\
\frac{1}{2(1+\rho)} & \frac{1}{2(1+\rho)} & 0 & \frac{\rho}{1+\rho} & 0 & \ldots \\
\frac{1}{3(1+\rho)} & 0 & \frac{2}{3(1+\rho)}  & 0 & \frac{\rho}{1+\rho}  & \ldots \\
\vdots & \vdots &  & &  & \ldots \\
\end{pmatrix}
$$
The non-null coefficients of the matrix $\mathcal{A}$ are given by $\mathcal{A}_{-1,-1}=1$ and for $n \geq 0$
\begin{eqnarray*}
\mathcal{A}_{n,-1} &=& \frac{1}{(n+1)(1+\rho)},\\
\mathcal{A}_{n,n-1} &=& \frac{n}{(n+1)(1+\rho)},\\
\mathcal{A}_{n,n+1} &=& \frac{\rho}{1+\rho}.
\end{eqnarray*}

Let us now consider the sub-matrix $A$ of $\mathcal{A}$ obtained by deleting the first row and the first column of matrix $\mathcal{A}$. The non-null coefficients of matrix $A$ are given for for $n \geq 0$ by
$$
a_{n,n+1}= \frac{\rho}{1+\rho} \mbox{ and } a_{n,n-1}= \frac{n}{(n+1)(1+\rho)},
$$
with the convention $a_{-1,0}=0$. This matrix is tridiagonal and sub-stochastic, and gives the transition probabilities of the Markov chain $(N_k)$ before absorption.

Let $e_n$ be the column vector with all entries equal to 0 except the $n$-th one equal to 1. Then, for $m \geq 0$,
$$
\P(N_k = m~|~ N_0=n) = {}^te_nA^{k} e_m
$$
where ${}^te_n$ is the row vector equal to the transpose of the column vector $e_n$. The probability that the tagged customer leaves the system at stage $k\geq 1$ and leaves $m$ customers in the queue is equal to
$$
\frac{1}{(m+1)(1+\rho)} {}^te_nA^{k-1} e_m
$$

Let $\nu$ denote the number of customers left in the system by the tagged customer upon service completion. We have
\begin{multline}
\label{eqnu}
\P(\nu=m~|~N_0=n) = \\ \frac{1}{(m+1)(1+\rho)}\sum_{k=1}^\infty  {}^te_nA^{k-1} e_m = \frac{1}{(m+1)(1+\rho)} {}^te_n(\I-A)^{-1} e_m .
\end{multline}
Similarly, if we denote by $\kappa$ the time at which the tagged customer leaves the system, we have for $k \geq 1$
$$
\P(\kappa = k ~|~N_0=n) = \sum_{m=0}^\infty  \frac{1}{(m+1)(1+\rho)} {}^te_nA^{k-1} e_m .
$$
Finally, for $k \geq 1$ and $m \geq 0$
$$
\P(\kappa = k, \nu= m ~|~N_0=n) =\frac{1}{(m+1)(1+\rho)} {}^te_nA^{k-1} e_m .
$$

Let $\alpha$ and $\delta$ respectively denote the number of arrivals (excluding the tagged customer) and departures, while  the tagged customer is in the system. Assume that there {are} $n$ customers in the system at the arrival time of the tagged customer and $m$ customers in the system upon service completion of the tagged customer. We have
$$
\alpha + \delta = \kappa-1 \mbox{  and  } \alpha - \delta = \nu -n, 
$$
so that 
$$
\alpha = \frac{\kappa+\nu-n-1}{2} \mbox{  and   } \delta = \frac{\kappa - \nu +n-1}{2}.
$$

In view of Equation~\eqref{eqnu}, we see that for characterizing the distribution of the various random variables introduced above, we have to compute the resolvent of the infinite matrix $A$, namely $(z\I-P)^{-1}$ as well as the powers of matrix $P$. For this purpose, we prove in the next section that this matrix induces a selfadjoint operator in some ad-hoc Hilbert space.

\section{Spectral properties of matrix $A$}
\label{spectral}

\subsection{Selfadjointness properties}

We introduce the same real Hilbert space as in~\cite{FGBJ}. Let 
$$
H=\left\{f\in \R^\N : \sum_{n=0}^\infty f_n^2 \pi_n<\infty \right\},
$$
where $\pi_n = (n+1)\rho^n$. The Hilbert space $H$ is equipped with the scalar product
$$
(f,g) = \sum_{n=0}^\infty f_n g_n \pi_n
$$
and the norm
$$
\| f\| = \sqrt{\sum_{n=0}^\infty f_n^2 \pi_n}.
$$

The infinite matrix $A$ induces in $H$ an operator that we also denote by $A$. By using the same arguments as in~\cite{FGBJ}, we can easily prove the following lemma, where we use the norm of the operator $A$ defined by
$$
\|A\| = \sup_{f\in H : \|f\|<1}|(Af,f)|;
$$
by definition, the operator $A$ is bounded if $\|A\|<\infty$.

\begin{lemma}
The operator $A$ is symmetric and bounded in $H$ and hence self-adjoint.
\end{lemma}
\begin{proof}
The symmetry of $A$ is straightforward since it is easily checked for every $f,g\in H$, $(Af,g) = (f,Ag)$ owing to the reversibility property $a_{n,n+1}\pi_n = a_{n,n-1}\pi_{n+1}$ for all $n \geq 0$.

For $f\in H$, we have by Schwarz inequality
\begin{eqnarray*}
\left| (Af,f) \right| &=&\frac{2}{1+\rho}\left|  \sum_{n=0}^\infty (n+1)\rho^{n+1}f_nf_{n+1} \right| \\
&\leq & \frac{2}{1+\rho} \sqrt{\sum_{n=0}^\infty (n+1)\rho^{n+1}f^2_n}\sqrt{ \sum_{n=0}^\infty    (n+1)\rho^{n+1}      f^2_{n+1} } \\ 
&\leq & \frac{2\sqrt{\rho}}{1+\rho} \|f\|^2.
\end{eqnarray*}
This implies that $\|A\|\leq \frac{2\sqrt{\rho}}{1+\rho} <1$.
\end{proof}

\subsection{Spectrum}

The spectrum $\sigma(A)$ of the operator $A$ is defined by
$$
\sigma(A) = \{z \in \R : (zI-A) \mbox{ is not invertible}\}.
$$
Since $\|A\|\leq \frac{2\sqrt{\rho}}{1+\rho} $, we know that $\sigma(A) \subset [-\frac{2\sqrt{\rho}}{1+\rho} ,\frac{2\sqrt{\rho}}{1+\rho} ]$. 

Let us consider some $f \in H$ such that $Af=x f$ for some real number $x$. By setting without loss of generality $f_{-1}=0$ and $f_0=1$, we have for $n \geq 0$
\begin{equation}
\label{recurQ}
(n+1) \rho f_{n+1} -(n+1) (1+\rho) x f_n + n f_{n-1}=0,
\end{equation}
which can be rewritten as 
$$
(n+1) (\sqrt{\rho})^{n+1} f_{n+1} -2(n+1)\frac{1+\rho}{2\sqrt{\rho}} x (\sqrt{\rho})^n f_n + n (\sqrt{\rho})^{n-1} f_{n-1}=0
$$

Let us introduce the Pollaczek polynomials $(P_n(x;a,b))$ defined for real $b$ and $a\geq |b|$ by the recursion: $P_{-1}(x;a,b)=0$, $P_0(x;a,b)=1$ and for $n \geq 0$
$$
(n+1) P_{n+1}(x;a,b)-((2n+1+a)x+b)P_n(x;a,b) + n P_{n-1}(x;a,b)=0.
$$
It is then easily checked that 
$$
f_n = \frac{1}{(\sqrt{\rho})^n}P_n\left( \frac{1+\rho}{2\sqrt{\rho}} x   ;1,0\right).
$$
In the following, we introduce the vectors $Q(x)$ such that the $n$-th component is
$$
Q_n(x) = \frac{1}{(\sqrt{\rho})^n}P_n\left( \frac{1+\rho}{2\sqrt{\rho}} x   ;1,0\right).
$$

For $x\in [-1,1]$, let $\theta \in [0,\pi]$ be such that $x = \cos\theta$ and define
$$
\tau= \frac{a \cos\theta +b}{2\sin \theta}.
$$
The Pollaczek polynomials $P_n(x;a,b)$ have the generating function defined for $x =\cos\theta$ by
$$
\sum_{n=0}^\infty P_n(x;a,b) z^n =(1-z e^{i\theta})^{-\frac{1}{2}+i \tau} (1-z e^{-i\theta})^{-\frac{1}{2}-i \tau} \stackrel{def}{=}\mathcal{P}(\theta,z);
$$
these polynomials are orthogonal with respect to the weight function supported by the interval $[-1,1]$ and given by
$$
w(x;a,b) = \frac{1}{2\cosh \pi\tau}e^{(2\theta -\pi)\tau},
$$
such that 
\begin{equation}
    \label{normcond}
\int_{-1}^1 P_n(x;a,b)P_m(x;a,b) w(x;a,b)dx = \frac{1}{2n+1+a}\delta_{m,n}
\end{equation}
where $\delta_{m,n}$ is the Kronecker symbol.

It follows that the polynomials $Q_n(x)$ are orthogonal with respect to the  measure $d\psi(x)$ given by
\begin{equation}
\label{defpsi}
\frac{d\psi}{dx} (x) = \frac{1+\rho}{\sqrt{\rho}} w\left( \frac{1+\rho}{2\sqrt{\rho}} x;1,0 \right) 
\end{equation}
or equivalently for $x=\frac{2\sqrt{\rho}}{1+\rho}\cos\theta$ with $\theta\in [0,\pi]$
$$
\frac{d\psi}{d\theta} =\frac{\sin \theta}{\cosh(\frac{\pi\cot\theta}{2})}\exp\left(\cot\theta\left(-\frac{\pi}{2}+\theta\right) \right),
$$
and satisfy
\begin{equation}
\label{orthorel}
\int_{-\frac{2\sqrt{\rho}}{1+\rho}}^{\frac{2\sqrt{\rho}}{1+\rho}}  Q_n(x) Q_m(x) d\psi(x) = \frac{1}{\pi_n}\delta_{m,n}.
\end{equation}
The polynomials $(Q_n(x))$ have the generating function
$$
\mathcal{Q}(x;z) = \sum_{n=0}^\infty Q_n(x) z^n =\left(1-\frac{z}{\sqrt{\rho}} e^{i\theta   }\right)^{-\frac{1}{2}+i \frac{\cot\theta}{2}} \left(1-\frac{z}{\sqrt{\rho}} e^{-i\theta}\right)^{-\frac{1}{2}-i \frac{\cot\theta}{2}} 
$$
for $x=\frac{2\sqrt{\rho}}{1+\rho}\cos\theta$ with $\theta\in [0,\pi]$. The generating function $\mathcal{Q}(x;z) $ satisfies the first order differential equation
$$
(z^2 -(1+\rho)x z +\rho)\frac{\partial \mathcal{Q}}{\partial z}+(z-(1+\rho)x) \mathcal{Q}=0.
$$

 It is worth noting that for real $z<1$
\begin{equation}
  \label{defPgen}  
  \mathcal{P}(\theta,z) = \frac{1}{\sqrt{1-2z\cos\theta+z^2}} 
  e^{\cot \theta \arctan\left( \frac{z \sin\theta}{1-z\cos\theta}\right)   }.
\end{equation}
In particular,
\begin{equation}
\label{Pgen0}
  \lim_{\theta \to 0}  \mathcal{P}(\theta,z) =  \frac{1}{1-z} e^{\frac{z}{1-z}} =  \mathcal{P}(0,z).
\end{equation}

With the above observations, we state the following lemma.

\begin{lemma}
The operator $A$ has the continuous spectrum $\left[-\frac{2\sqrt{\rho}}{1+\rho},\frac{2\sqrt{\rho}}{1+\rho}\right]$ and spectral measure $d\psi(x)$ defined by Equation~\eqref{defpsi}.
\end{lemma}

\section{Computation of distributions}
\label{distributions}

Let us first consider the random variables $\nu$ and $\kappa$.

\begin{proposition}
The distribution of the random variable $\nu$ (i.e., the number of jobs in the system upon completion of the tagged job) conditionally on the number of jobs in the system upon arrival of the tagged job is given for $m \geq 0$ by
\begin{equation}
\label{numn}
\P(\nu =m~|~N_0=n) =  \frac{\rho^m}{1+\rho} \int_{-\frac{2\sqrt{\rho}}{1+\rho}}^{\frac{2\sqrt{\rho}}{1+\rho}} \frac{Q_m(x) Q_n(x)}{1-x}d\psi(x).
\end{equation}
Moreover,
\begin{equation}
\label{kappan}
\P(\kappa = k ~|~N_0=n) = \frac{1}{1+\rho} \int_{-\frac{2\sqrt{\rho}}{1+\rho}}^{\frac{2\sqrt{\rho}}{1+\rho}} x^{k-1} \mathcal{Q}(x;\rho)Q_n(x)d\psi(x)
\end{equation}
and
\begin{equation}
\label{kappanun}
\P(\kappa = k,\nu=m ~|~N_0=n) = \frac{\rho^m}{(1+\rho)} \int_{-\frac{2\sqrt{\rho}}{1+\rho}}^{\frac{2\sqrt{\rho}}{1+\rho}} x^{k-1} Q_m(x) Q_n(x)d\psi(x)
\end{equation}
\end{proposition}

\begin{proof}
By using Equation~\eqref{eqnu}, we have
$$
\P(\nu =m~|~N_0=n) = \frac{1}{(m+1)(1+\rho)\pi_n}\left(e_n,(\I-A)^{-1}e_m\right).
$$
By using the spectral identity, we have
$$
\left(e_n,(\I-A)^{-1}e_m\right) = \int_{-\frac{2\sqrt{\rho}}{1+\rho}}^{\frac{2\sqrt{\rho}}{1+\rho}} \frac{1}{1-x} (e_n,(e_m)_x ) d\psi(x) ,
$$
where $(e_m)_x$ is the projection of the vector $e_m$ on the vector space spanned by the vector $Q(x)$. By using the orthogonality relation~\eqref{orthorel}, we have
$$
e_m = \pi_m \int_{-\frac{2\sqrt{\rho}}{1+\rho}}^{\frac{2\sqrt{\rho}}{1+\rho}} Q_m(x) Q(x)d\psi(x).
$$
Hence,
$$
(e_m)_x= \pi_m Q_m(x) Q(x)
$$
and Equation~\eqref{numn} easily follows.

Similarly, we have
\begin{eqnarray*}
\P(\kappa = k ~|~N_0=n) &=& \sum_{m=0}^\infty  \frac{1}{(m+1)(1+\rho)\pi_n} (e_n,A^{k-1}e_m) \\
&=& \sum_{m=0}^\infty \frac{1}{(m+1)(1+\rho)\pi_n} \int_{-\frac{2\sqrt{\rho}}{1+\rho}}^{\frac{2\sqrt{\rho}}{1+\rho}} (e_n,x^{k-1} Q(x))d\psi(x) \\
&=& \sum_{m=0}^\infty \frac{\rho^m}{(1+\rho)} \int_{-\frac{2\sqrt{\rho}}{1+\rho}}^{\frac{2\sqrt{\rho}}{1+\rho}} x^{k-1} Q_m(x)Q_n(x)d\psi(x)
\end{eqnarray*}
and Equation~\eqref{kappan} follows.
Finally,
$$
\P(\kappa = k, \nu= m ~|~N_0=n) =\frac{1}{(m+1)(1+\rho)\pi_n} (e_n,A^{k-1} e_m )
$$
and Equation~\eqref{kappanun} easily follows.
\end{proof}

It is worth noting that 
\begin{multline*}
\sum_{m=0}^\infty \P(\nu =m~|~N_0=n) = \frac{1}{1+\rho} \int_{-\frac{2\sqrt{\rho}}{1+\rho}}^{\frac{2\sqrt{\rho}}{1+\rho}} \frac{\mathcal{Q}(x;\rho) Q_n(x)}{1-x}d\psi(x)\\
= \int_0^\pi   \frac{P_n(\cos\theta;1,0) \left(1-{\sqrt{\rho}} e^{i\theta   }\right)^{-\frac{1}{2}+i \frac{\cot\theta}{2}} \left(1-{\sqrt{\rho}} e^{-i\theta}\right)^{-\frac{1}{2}-i \frac{\cot\theta}{2}} }{(1-2\sqrt{\rho}\cos \theta +\rho) \sqrt{\rho}^n} d\psi(\theta).
\end{multline*}
Let $\phi(\theta)$ defined by
$$
\phi(\theta) = \atan{\frac{2\sqrt{\rho}\sin\theta}{1-2\sqrt{\rho}\cos\theta}}
$$
so that 
$$
1-\sqrt{\rho} e^{i\theta}=  \sqrt{1+\rho-2\sqrt{\rho}\cos \theta} e^{-i\phi(\theta)}.
$$
We then have
\begin{multline*}
\sum_{m=0}^\infty \P(\nu =m~|~N_0=n) = \\
\frac{1}{(\sqrt{\rho})^n} \int_0^\pi   \frac{\sin \theta  P_n(\cos\theta;1,0) }{(1-2\sqrt{\rho}\cos \theta +\rho)^{\frac{3}{2}} \cosh(\frac{\pi\cot\theta}{2})}\exp\left(\cot\theta\left(-\frac{\pi}{2}+\theta +\phi(\theta)\right) \right)d\theta.
\end{multline*}
By introducing the Laplace transform $W_n^\star(z)$ of the random variable $W_n$, which is the sojourn time of a tagged job entering the $M/M/1$ PS queue while there are $n$ jobs in service, we have from~\cite{FGBJ}
$$
\sum_{m=0}^\infty \P(\nu =m~|~N_0=n) =W_n^\star (0)=1, 
$$
as expected. This shows in particular that for all $n \geq 0$
\begin{equation}
\label{eqtech1}
\frac{1}{1+\rho} \int_{-\frac{2\sqrt{\rho}}{1+\rho}}^{\frac{2\sqrt{\rho}}{1+\rho}} \frac{\mathcal{Q}(x;\rho) Q_n(x)}{1-x}d\psi(x) = 1.
\end{equation}

By using the above computations, we have the following corollary\added{.}
\begin{corollary}
When the tagged job enters an $M/M/1$-PS queue in the stationary regime, the number $\nu$ of jobs left in the system upon service completion of the tagged customer has distribution
\begin{equation}
\label{stationarymu}
\P(\nu = m) = (1-\rho)\rho^m.
\end{equation}

The generating function of the random variable $\kappa$ (i.e., the time at which the tagged customer leaves the system) is given by
\begin{equation}
\label{genkappa}
\kappa^\star(z) = \sum_{k=0}^\infty \P(\kappa = k) z^k = \frac{1-\rho}{1+\rho} \int_{-\frac{2\sqrt{\rho}}{1+\rho}}^{\frac{2\sqrt{\rho}}{1+\rho}} \frac{z}{1-zx} \mathcal{Q}(x;\rho)^2 d\psi(x) .
\end{equation}
\end{corollary}

\begin{proof}
Since $\P(N_0=n)=(1-\rho)\rho^n$, we have
$$
\P(\nu = m) = \frac{(1-\rho) \rho^m}{1+\rho} \int_{-\frac{2\sqrt{\rho}}{1+\rho}}^{\frac{2\sqrt{\rho}}{1+\rho}} \frac{\mathcal{Q}(x;\rho) Q_m(x)}{1-x}d\psi(x) = (1-\rho)\rho^m,
$$
where we have used Equation~\eqref{eqtech1} in the last step.

Using again  $\P(N_0=n)=(1-\rho)^n$, we immediately obtain Equation~\eqref{genkappa}. We verify that $\kappa^\star(1) =1$ by using Equation~\eqref{eqtech1}.
\end{proof}

Equation~\eqref{stationarymu} is consistent with the fact that in the stationary regime, the distribution of the occupancy of the queue seen by departing customers is the same as the  distribution seen by arriving customers (equal to the stationary distribution owing to the PASTA property); this is a classical result in queuing theory. Moreover, the sojourn time of the tagged customer finding $n$ customers in the queue is 
$$
W_n = \mathcal{E}_1 + \ldots +\mathcal{E}_\kappa,
$$
where $(\mathcal{E}_i)$ is a sequence of independent and identically distributed exponential random variables with mean $\frac{1}{1+\rho}$. It follows that the Laplace transform of $W_n$ is
\begin{eqnarray*}
W^\star_n(z)  &=& \sum_{k=1}^\infty \frac{1}{1+\rho}    \int_{-\frac{2\sqrt{\rho}}{1+\rho}}^{\frac{2\sqrt{\rho}}{1+\rho}} x^{k-1} \left(\frac{1+\rho}{z+1+\rho}  \right)^k \mathcal{Q}(x;\rho)Q_n(x)d\psi(x) \\
&=& \int_{-\frac{2\sqrt{\rho}}{1+\rho}}^{\frac{2\sqrt{\rho}}{1+\rho}} \frac{1}{z+(1+\rho)(1-x)}  \mathcal{Q}(x;\rho)Q_n(x)d\psi(x) \\
&=&\frac{1}{(\sqrt{\rho})^n} \int_0^\pi   \frac{\sin \theta  P_n(\cos\theta;1,0) \exp\left(\phi(\theta)\cot\theta \right) }{(1-2\sqrt{\rho}\cos \theta +\rho)^{\frac{1}{2}}   (z+1+\rho-2\sqrt{\rho}\cos\theta)} d\psi(\theta).
\end{eqnarray*}
By inverting the Laplace transform we eventually obtain
\begin{multline*}
\P(W_n >y) =  \\ \frac{1}{(\sqrt{\rho})^n} \int_0^\pi   \frac{\sin \theta  P_n(\cos\theta;1,0)  \exp\left(\phi(\theta) \cot\theta \right) }{(1-2\sqrt{\rho}\cos \theta +\rho)^{\frac{3}{2}} }  e^{-(1+\rho-2\sqrt{\rho}\cos\theta)y } d\psi(\theta),
\end{multline*}
which corresponds to Equation~(15) in~\cite{FGBJ}.

\begin{corollary}
The number of arrivals $\alpha$ at the queue while the tagged job  is in service is given by
\begin{equation}
\label{disalpha}
\P(\alpha=j) =\frac{1-\rho}{1+\rho} \sum_{n=0}^\infty \sum_{m=0}^{j+n} \rho^{m+n}  \int_{-\frac{2\sqrt{\rho}}{1+\rho}}^{\frac{2\sqrt{\rho}}{1+\rho}} x^{2j+n-m}Q_m(x)Q_n(x)d\psi(x).
\end{equation}
and the number of departures $\delta$ under the same condition is such that $\delta  \stackrel{d}{=} \alpha$ (equality in distribution).
\end{corollary}

\begin{proof}
Let $ j\geq 0$ be an integer. We have 
\begin{eqnarray*}
\P(\alpha=j~|~N_0=n) &=& \P(\kappa+\nu-n-1=2j~|~N_0=n) \\
&=& \sum_{m=0}^{2j+n}\P(\kappa=2j+n-m+1,\nu=m~|~N_0=n) \\
&=& \sum_{m=0}^{2j+n} \frac{1}{(m+1)(1+\rho)} {}^te_n A^{2j+n-m}e_m  \\ 
&=& \sum_{m=0}^{2j+n} \frac{1}{(m+1)(1+\rho)\pi_n} (e_n ,A^{2j+n-m}e_m ) \\
&=&\sum_{m=0}^{2j+n} \frac{\pi_m}{(m+1)(1+\rho)}  \int_{-\frac{2\sqrt{\rho}}{1+\rho}}^{\frac{2\sqrt{\rho}}{1+\rho}} x^{2j+n-m}Q_m(x)Q_n(x)d\psi(x)
\end{eqnarray*}

We note that $x^{2j+n-m}Q_n(x)$ is a polynomial with degree $2j+2n-m$ and hence owing to the orthogonality property of polynomials $Q_m(x)$
\begin{equation}
\label{prodtech}
\int_{-\frac{2\sqrt{\rho}}{1+\rho}}^{\frac{2\sqrt{\rho}}{1+\rho}} x^{2j+n-m}Q_m(x)Q_n(x)d\psi(x) =0
\end{equation}
for $2j+2n-m < m$ that is $m > j+n$. It follows that 
$$
\P(\alpha=j~|~N_0=n) =\sum_{m=0}^{j+n} \frac{\rho^m}{1+\rho}  \int_{-\frac{2\sqrt{\rho}}{1+\rho}}^{\frac{2\sqrt{\rho}}{1+\rho}} x^{2j+n-m}Q_m(x)Q_n(x)d\psi(x)
$$
and Equation~\eqref{disalpha} follows by deconditioning on $N_0$.

By using similar arguments, we have
\begin{eqnarray*}
\P(\delta=j~|~N_0=n) &=& \P(\kappa-\nu+n-1=2j~|~N_0=n) \\
&=& \frac{1}{1+\rho}  \sum_{m=(n-2j)^+}^{\infty} \rho^m \int_{-\frac{2\sqrt{\rho}}{1+\rho}}^{\frac{2\sqrt{\rho}}{1+\rho}} x^{2j-n+m}Q_m(x)Q_n(x)d\psi(x)
\end{eqnarray*}
By using Equation~\eqref{prodtech}, we deduce that
$$
\int_{-\frac{2\sqrt{\rho}}{1+\rho}}^{\frac{2\sqrt{\rho}}{1+\rho}} x^{2j-n+m}Q_m(x)Q_n(x)d\psi(x)=0
$$
for $m <j+n$ and hence
\begin{eqnarray*}
\P(\delta=j) &=&\frac{1-\rho}{1+\rho} \sum_{n=0}^\infty \sum_{m=(n-j)^+}^{\infty} \rho^{m+n}  \int_{-\frac{2\sqrt{\rho}}{1+\rho}}^{\frac{2\sqrt{\rho}}{1+\rho}} x^{2j-n+m}Q_m(x)Q_n(x)d\psi(x) \\
&=&\frac{1-\rho}{1+\rho} \sum_{m=0}^\infty \sum_{n=0}^{m+j} \rho^{m+n}  \int_{-\frac{2\sqrt{\rho}}{1+\rho}}^{\frac{2\sqrt{\rho}}{1+\rho}} x^{2j-n+m}Q_m(x)Q_n(x)d\psi(x)\\
&=& \P(\alpha=j),
\end{eqnarray*}
so that we have $\delta  \stackrel{d}{=}\alpha$.
\end{proof}

To conclude this section, let us study the asymptotic behavior of the common distribution of the random variables $\alpha$ and $\delta$. In the following, we set $e=\exp(1)$.

\begin{proposition}
When $j$ tends to infinity, we have
\begin{equation}
    \label{alphaj}
\P(\delta=j) \sim\frac{4 }{1-{\rho}}e^{2 \frac{1+\rho}{1- {\rho}}} \sqrt{\frac{8}{3}\left(\frac{\pi}{2}\right)^{\frac{5}{3}}} r(j),
\end{equation}
where
\begin{equation}
    \label{defrj}
r(j) =     \frac{1}{j^{\frac{5}{6}}}    e^{-3 \left(\frac{\pi}{2}  \right)^{\frac{2}{3}} j^{\frac{1}{3}}     }   \left(\frac{2\sqrt{\rho}}{1+\rho}\right)^{2j} .
\end{equation}
\end{proposition}

\begin{proof}
The idea of the proof is to show that when $j$ tends to infinity 
$$
\P(\delta= j) \sim  \frac{1-\rho}{1+\rho} \sum_{m=0}^\infty \sum_{n=0}^{\infty} \rho^{m+n}  \int_{-\frac{2\sqrt{\rho}}{1+\rho}}^{\frac{2\sqrt{\rho}}{1+\rho}} x^{2j-n+m}Q_m(x)Q_n(x)d\psi(x)
$$
and to use the same technique as in~\cite{Flatto} to obtain Equation~\eqref{alphaj}.

By definition, we have
$$
\P(\delta =j) =\frac{1-\rho}{1+\rho} \sum_{n=0}^\infty \sum_{m=0}^{j+n} (\sqrt{\rho})^{m+n} \left(\frac{2\sqrt{\rho}}{1+\rho}\right)^{2j+n-m} \int_{0}^{\pi}(\cos\theta)^{2j}h_{n,m}(\theta)d\psi(\theta),
$$
where 
$$
h_{n,m}(\theta) = (\cos\theta)^{n-m}P_m(\cos\theta;1,0)P_n(\cos\theta;1,0).
$$
We clearly have
$$
 \left(\frac{2\sqrt{\rho}}{1+\rho}\right)^{2j}\int_{\pi}^{0}( \cos\theta)^{2j}h_{n,m}(\theta)d\psi(\theta) = 2 \left(\frac{2\sqrt{\rho}}{1+\rho}\right)^{2j}\int_{0}^{\frac{\pi}{2}}( \cos\theta)^{2j}h_{n,m}(\theta)d\psi(\theta),
$$

Let $\eta\in \left(0,\frac{\pi}{2}\right)$. We can write
$$
\P(\alpha=j) = p_1(j) +p_2(j),
$$
where
\begin{eqnarray*}
p_1(j) &=&2 \frac{1-\rho}{1+\rho} \sum_{n=0}^\infty \sum_{m=0}^{j+n} (\sqrt{\rho})^{m+n} \left(\frac{2\sqrt{\rho}}{1+\rho}\right)^{2j+n-m} \int_{0}^{\eta}(\cos\theta)^{2j}h_{n,m}(\theta)d\psi(\theta),\\
p_2(j) &=& 2\frac{1-\rho}{1+\rho} \sum_{n=0}^\infty \sum_{m=0}^{j+n} (\sqrt{\rho})^{m+n} \left(\frac{2\sqrt{\rho}}{1+\rho}\right)^{2j+n-m} \int_{\eta}^{\frac{\pi}{2}}(\cos\theta)^{2j}h_{n,m}(\theta)d\psi(\theta).
\end{eqnarray*}

We first note that 
\begin{align*}
&2\left|\int_{\eta}^{\frac{\pi}{2}}(\cos\theta)^{2j+n-m}  P_m(\cos\theta;1,0)P_n(\cos\theta;1,0)d\psi(\theta)  \right|  \\
&\leq (\cos\eta)^{2j+n-m} \int_{\eta}^{{\pi}}\left|P_m(\cos\theta;1,0)P_n(\cos\theta;1,0)\right|d\psi(\theta) \\
&\leq (\cos\eta)^{2j+n-m} \sqrt{\int_0^\pi P_m(\cos\theta;1,0)^2d\psi(\theta)} \sqrt{\int_{0}^{\pi}P_n(\cos\theta;1,0)^2d\psi(\theta)} \\
& \leq \frac{(\cos\eta)^{2j+n-m}}{\sqrt{(m+1)(n+1)}},
\end{align*}
where we have used the normalizing condition~\eqref{normcond}. It follows that if we choose $\eta$ sufficiently small so that $\cos\eta>\frac{1+\rho}{2}$
\begin{eqnarray*}
p_2(j) &\leq&  \frac{1-\rho}{1+\rho} \sum_{n=0}^\infty \left(\frac{2\rho\cos\eta}{1+\rho}\right)^n \sum_{m=0}^\infty \left(\frac{1+\rho}{2 \cos\eta}\right)^m \left(\frac{2\sqrt{\rho}\cos\eta}{1+\rho}\right)^{2j}\\
&=& \frac{2(1-\rho)\cos\eta}{(1+\rho-2\rho \cos\eta)(2\cos\eta - 1- \rho)} \left(\frac{2\sqrt{\rho}\cos\eta}{1+\rho}\right)^{2j}.
\end{eqnarray*}
It is then  easily checked that 
$$
\lim_{j\to \infty}\frac{p_2(j)}{r(j)}=0,
$$
where $r(j)$ is defined by Equation~\eqref{defrj}.

Now, we can decompose $p_1(j)$ as $p_1(j) = p_{1,1}(j)-p_{1,2}(j)$ with
$$
 p_{1,1}(j) =  2 \frac{1-\rho}{1+\rho} \sum_{n=0}^\infty \sum_{m=0}^{\infty} (\sqrt{\rho})^{m+n} \left(\frac{2\sqrt{\rho}}{1+\rho}\right)^{2j+n-m}\int_{0}^{\eta}(\cos\theta)^{2j}h_{n,m}(\theta)d\psi(\theta)*
 $$
 and
 $$
  p_{1,2}(j)  = 2 \frac{1-\rho}{1+\rho} \sum_{n=0}^\infty \sum_{m=n+j+1}^{\infty} (\sqrt{\rho})^{m+n} \left(\frac{2\sqrt{\rho}}{1+\rho}\right)^{2j+n-m}\int_{0}^{\eta}(\cos\theta)^{2j}h_{n,m}(\theta)d\psi(\theta).
$$
By using the same arguments as above, we can show that the quantity $p_{1,2}(j)$ is less than or equal to 
\begin{multline*}
  \frac{1-\rho}{1+\rho} \sum_{n=0}^\infty \sum_{m=n+j+1}^\infty  (\sqrt{\rho})^{m+n} \left(\frac{2\sqrt{\rho}}{1+\rho}\right)^{2j+n-m} \\ \int_{0}^{{\pi}}\left|P_m(\cos\theta;1,0)P_n(\cos\theta;1,0)\right|d\psi(\theta) 
 \end{multline*}
 and then
 \begin{eqnarray*}
 p_{1,2}(j)& \leq&  \frac{1-\rho}{1+\rho} \sum_{n=0}^\infty \sum_{m=n+j+1}^\infty  (\sqrt{\rho})^{m+n} \left(\frac{2\sqrt{\rho}}{1+\rho}\right)^{2j+n-m}\\
 &=& \frac{1}{1-\rho} \left(\frac{1+\rho}{2}\right)^j\left(\frac{2\sqrt{\rho}}{1+\rho}\right)^{2j}.
\end{eqnarray*}
Because $  \frac{1+\rho}{2}<1$, we clearly have 
$$
\lim_{j\to \infty}\frac{p_{1,2}(j)}{r(j)}=0.
$$

Finally, we have
$$
 p_{1,1}(j) =  2 \frac{1-\rho}{1+\rho}  \int_{0}^{\eta}\left(\frac{2\sqrt{\rho}\cos\theta}{1+\rho}\right)^{2j}      \mathcal{P}\left(\theta,\frac{2\sqrt{\rho}\cos\theta}{1+\rho}\right)  
 \mathcal{P}\left(\theta,\frac{1+\rho}{2\cos\theta}\right)  
 d\psi(\theta)
$$

For small $\theta$
$$
\frac{\sin \theta}{\cosh\frac{\pi \cot\theta}{2}}\exp\left(\left(\theta-\frac{\pi}{2}\right)\cot\theta\right) \sim 2e \theta \exp\left( -\frac{\pi}{\theta} \right) 
$$
and
$$
(\cos\theta)^n \sim 1-\frac{n\theta^2}{2}.
$$
By mimicking the proof in~\cite[Section~6]{Flatto}, we have for large $j$
\begin{eqnarray*}
p_{1,1}(j) &\sim& 2  \kappa_1   \left(\frac{2\sqrt{\rho}}{1+\rho}\right)^{2j}    \int_{0}^{\frac{\pi}{2}} e^{-{j\theta^2}-\frac{\pi}{\theta} }\theta d\theta\\  &=&   \kappa_1 \left(\frac{2\sqrt{\rho}}{1+\rho}\right)^{2j}  \frac{2}{j^{\frac{2}{3}}} \int_{0}^{\frac{\pi}{2}j^{\frac{1}{3}}} e^{-j^{\frac{1}{3}}\left(\theta^2+\frac{\pi}{\theta}\right) }\theta d\theta \\
&=&   \kappa_1 \left(\frac{2\sqrt{\rho}}{1+\rho}\right)^{2j}      \frac{1}{ j^{\frac{2}{3}}} \int_{0}^{\frac{\pi^2}{4}j^{\frac{2}{3}}} e^{-j^{\frac{1}{3}}\left(\varphi+\frac{\pi}{\sqrt{\varphi}}\right) }d\varphi
\end{eqnarray*}
where we have set $\varphi = \theta^2$ and 
$$
\kappa_1 = 2e \frac{1-\rho}{1+\rho}  \mathcal{P}\left(0,\frac{2{\rho}}{1+\rho}\right)  
 \mathcal{P}\left(0,\frac{1+\rho}{2}\right)
$$.

Let us introduce $g(\varphi) =\varphi+\frac{\pi}{\sqrt{\varphi}} $. We have
$$
g'(\varphi) = 1-\frac{\pi}{2\varphi^{\frac{3}{2}}}
$$
and $g'(\varphi)=0$ for $\varphi=\varphi_0 =\left(\frac{\pi}{2}\right)^{\frac{2}{3}} $. In addition, $g''(\varphi_0)= \frac{3}{2}\left(\frac{\pi}{2}\right)^{-\frac{2}{3}}$.
By using Laplace's method, we  obtain
$$
 \int_{0}^{\frac{\pi^2}{4}j^{\frac{2}{3}}} e^{-j^{\frac{1}{3}}\left(\varphi+\frac{\pi}{\sqrt{\varphi}}\right) }d\varphi \sim \sqrt{\frac{2\pi}{j^{\frac{1}{3}}g''(\varphi_0)}}e^{-j^{\frac{1}{3}}g(\varphi_0)} = \sqrt{\frac{8}{3j^{\frac{1}{3}}}\left(\frac{\pi}{2}\right)^{\frac{5}{3}}}e^{-3  \left(\frac{\pi}{2}  \right)^{\frac{2}{3}}j^{\frac{1}{3}}     }
$$
It follows that for large $j$
$$
p_{1,1}(j)(j) \sim \kappa_1  \left(\frac{2\sqrt{\rho}}{1+\rho}\right)^{2j}  \sqrt{\frac{8}{3}\left(\frac{\pi}{2}\right)^{\frac{5}{3}}}   \frac{1}{j^{\frac{5}{6}}}    e^{-3  \left(\frac{\pi}{2}  \right)^{\frac{2}{3}}  j^{\frac{1}{3}}   }
$$

By using Equation~\eqref{Pgen0}, we obtain
$$
\mathcal{P}\left(0,\frac{2{\rho}}{1+\rho}\right) = \frac{1+\rho}{1-\rho}e^{\frac{2{\rho}}{1-{\rho}}}, \quad 
 \mathcal{P}\left(0,\frac{1+\rho}{2}\right) =  \frac{2}{1-\rho} e^{\frac{1+\rho}{1-\rho}},
$$
and Equation~\eqref{alphaj} follows.
\end{proof}

 In the following section, we numerically compute the distribution of $\delta$ and we illustrate the asymptotic behavior of the tail of the distribution.


\section{Numerical Experiments}
\label{numerical}

\subsection{Number of departures}
To compute $\P(\delta=j)$ for $j \geq 0$, we use the following procedure. Since the polynomials $Q_n(x)$ satisfy the recursion defined by Equation~\eqref{recurQ} with  $Q_0(x) = 1$ and $Q_1(x) =\frac{1+\rho}{\rho} x$, we can represent the polynomial $Q_n(x)$ by a list  $L_n=\{ l_{n,n}, l_{n,n-1}, \ldots, l_{n,0} \}$ so that
$$
Q_n(x) = \sum_{k=0}^n l_{n,k} x^k,
$$
where the coefficients $l_{n,k}$ can be recursively computed by using the fact that for $n \geq 2$ and $0\leq k \leq n$
$$
l_{n,k} = \frac{1+\rho}{\rho}l_{n-1,k-1}\mathbbm{1}_{\{1\leq k\leq n}-\frac{n-1}{n}l_{n-2,k}\mathbbm{1}_{\{0\leq k\leq n-2}
$$
with $l_{0,0}=1$, $l_{1,1}= \frac{1+\rho}{\rho}$ and $l_{1,0}=0$.

By using the above definitions, the quantity $\P(\alpha=j)$ can be expressed as
$$
\P(\delta = j) = \frac{1-\rho}{1+\rho}\sum_{n=0}^{\infty} \sum_{m=0}^{j+n} \rho^{m+n}\sum_{s=0}^{m}l_{m,s}\sum_{t=0}^{n}l_{n,t}
\mu_{2j+n-m+s+t},
$$
where $\mu_n$ is the $n$-th moment of the measure $d\psi(x)$ defined by
\begin{equation}
\label{defmun}
\mu_n = \int_{-\frac{2\sqrt{\rho}}{1+\rho}}^{\frac{2\sqrt{\rho}}{1+\rho}} x^{n} d\psi(x).
\end{equation}
By using the orthogonality relation~\eqref{orthorel}, we have for $n \geq 1$
$$
\int_{-\frac{2\sqrt{\rho}}{1+\rho}}^{\frac{2\sqrt{\rho}}{1+\rho}} Q_n(x) d\psi(x) =0
$$
and the moments $\mu_n$ for $n \geq 1$ can be recursively computed by 
$$
\mu_n= -\frac{1}{l_{n,n}}\sum_{k=0}^{n-1}l_{n,k} \mu_k
$$
with $\mu_0=1$.

The moment $\mu_n$ is analytically defined by
$$
\mu_n = \int_{0}^{{\pi} } \left(\frac{2\sqrt{\rho}}{1+\rho}  \right)^n(\cos\theta)^n \frac{\sin \theta}{\cosh\frac{\pi \cot\theta}{2}}\exp\left(\left(\theta-\frac{\pi}{2}\right)\cot\theta\right)d\theta.
$$
It is easily checked that $\mu_n=0$ for odd $n$ and for even $n$
$$
\mu_n = 2\int_{0}^{\frac{\pi}{2}} \left(\frac{2\sqrt{\rho}}{1+\rho}  \right)^n(\cos\theta)^n \frac{\sin \theta}{\cosh\frac{\pi \cot\theta}{2}}\exp\left(\left(\theta-\frac{\pi}{2}\right)\cot\theta\right)d\theta
$$
We can estimate the value of $\mu_n$ for large $n$ by using the same arguments as above  and we obtain 
$$
\mu_n \sim 2e  \left(\frac{2\sqrt{\rho}}{1+\rho}  \right)^n \frac{1}{n^{\frac{5}{6}}} \sqrt{\frac{8\pi^{\frac{5}{3}}}{3}}e^{-\pi^{\frac{2}{3}}n^\frac{1}{3}}
$$
when $n$ tends to infinity.

To numerically compute the value of $\P(\delta=j)$, we truncate the sum for $n$ ranging from 0 to infinity. We fix some $N>0$ and we increase the value of $N$ so as to obtain $\sum_{j=0}^N \P(\alpha=j)=1-\varepsilon$ for some $\varepsilon \ll 1.$

Figure~\ref{fig:Palpha} illustrates the probability distribution function of the number departures $\delta$ while the tagged job is in service for various system loads. In this figure, we have indicated the value of $\varepsilon$ used to truncate the infinite sum in the expression of $\P(\delta=j)$.

\begin{figure}[hbtp]
	\centering
   \includegraphics[scale=0.45, trim=0 0 0 0, clip] {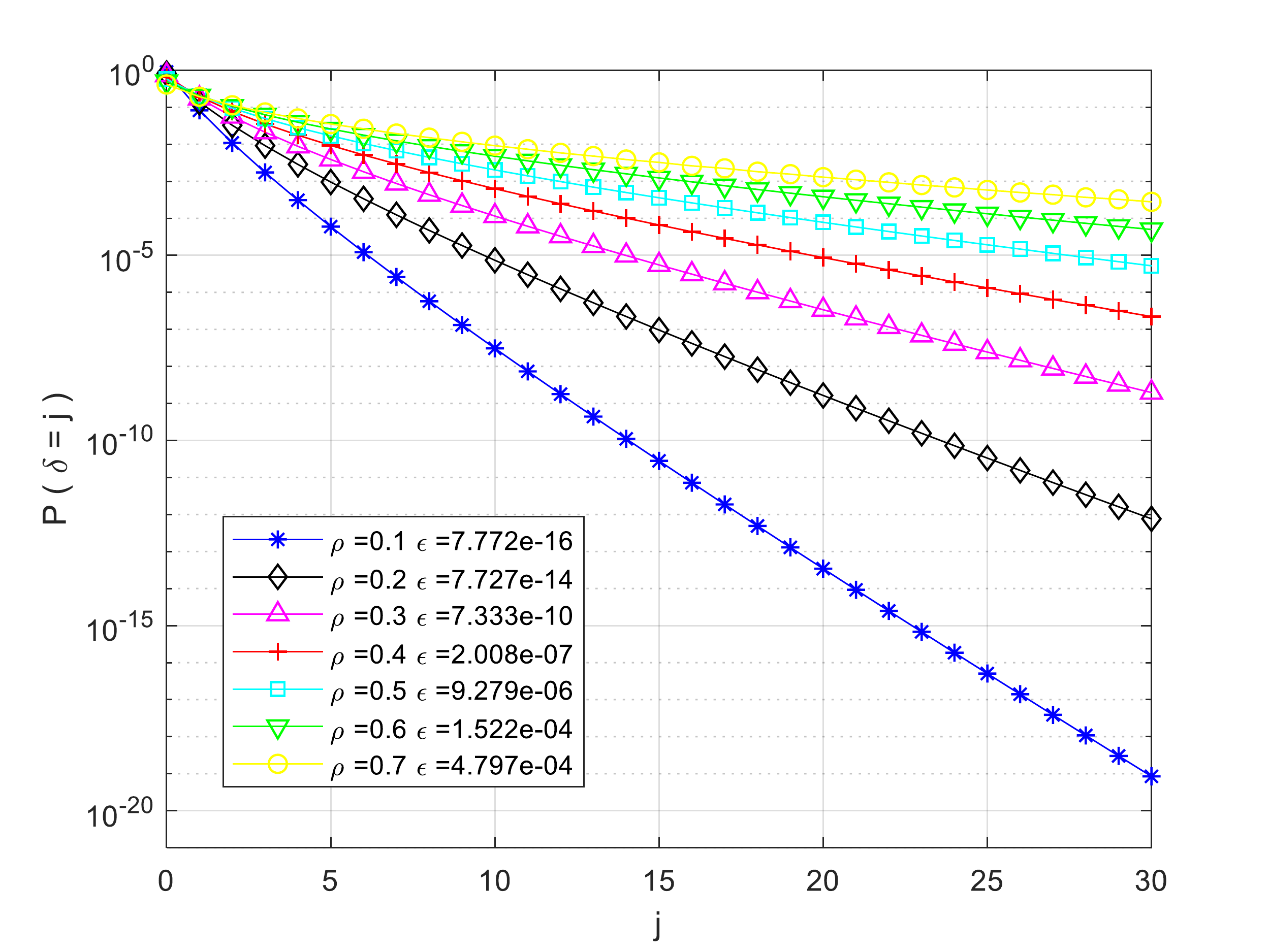}
    \caption{$\P(\delta=j)$ }
  \label{fig:Palpha}
\end{figure}

Figure~\ref{fig:Palpha_aprox} presents the asymptotic behavior given by Equation~\eqref{alphaj} (dashed line) of the distribution of the number of departures while the tagged customer is in service. The asymptotic formula is rapidly accurate for light load and the convergence is much slower for high loads. This phenomenon also occurs for the sojourn time distribution of a tagged customer by using the approximation established by Flatto~\cite{Flatto} for the Random Order Service queue and which is applicable to the PS queue modulo a factor $1/\rho$ (see~\cite{Borst2006} for details). 

\begin{figure}[hbtp]
	\centering
   \includegraphics[scale=0.45, trim=0 0 0 0, clip] {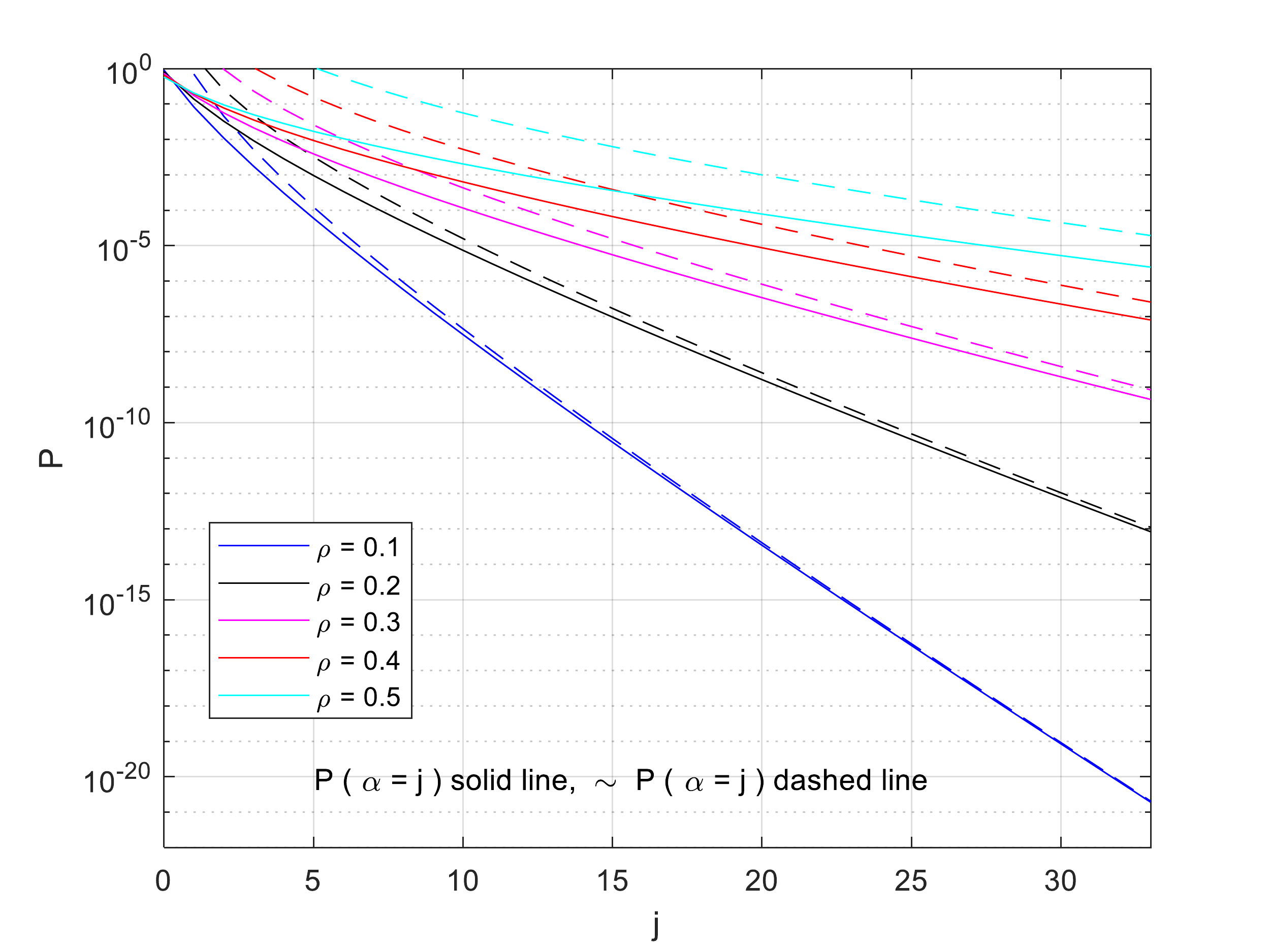}
      \caption{Asymptotic of the common distribution $\alpha$ and $\delta$.}
  \label{fig:Palpha_aprox}
\end{figure}

\subsection{Approximation}

As discussed in the Introduction, we now  consider the following approximation. Since the tagged customer has at each instant  a probability of completing service  equal to that of  any other customer present in the queue,  we may assume that the tagged customer is randomly served together with those customers in the residual busy period of the queue i.e., the period of activity  of the queue following the entrance of the tagged customer in the queue up to exhaustion. {See Figure~\ref{fig:busy_period} for an illustration of the residual busy periof and the associated random variables.}

\begin{figure}[hbtp]
	\centering
   \includegraphics[scale=0.45, trim=0 170 0 0, clip] {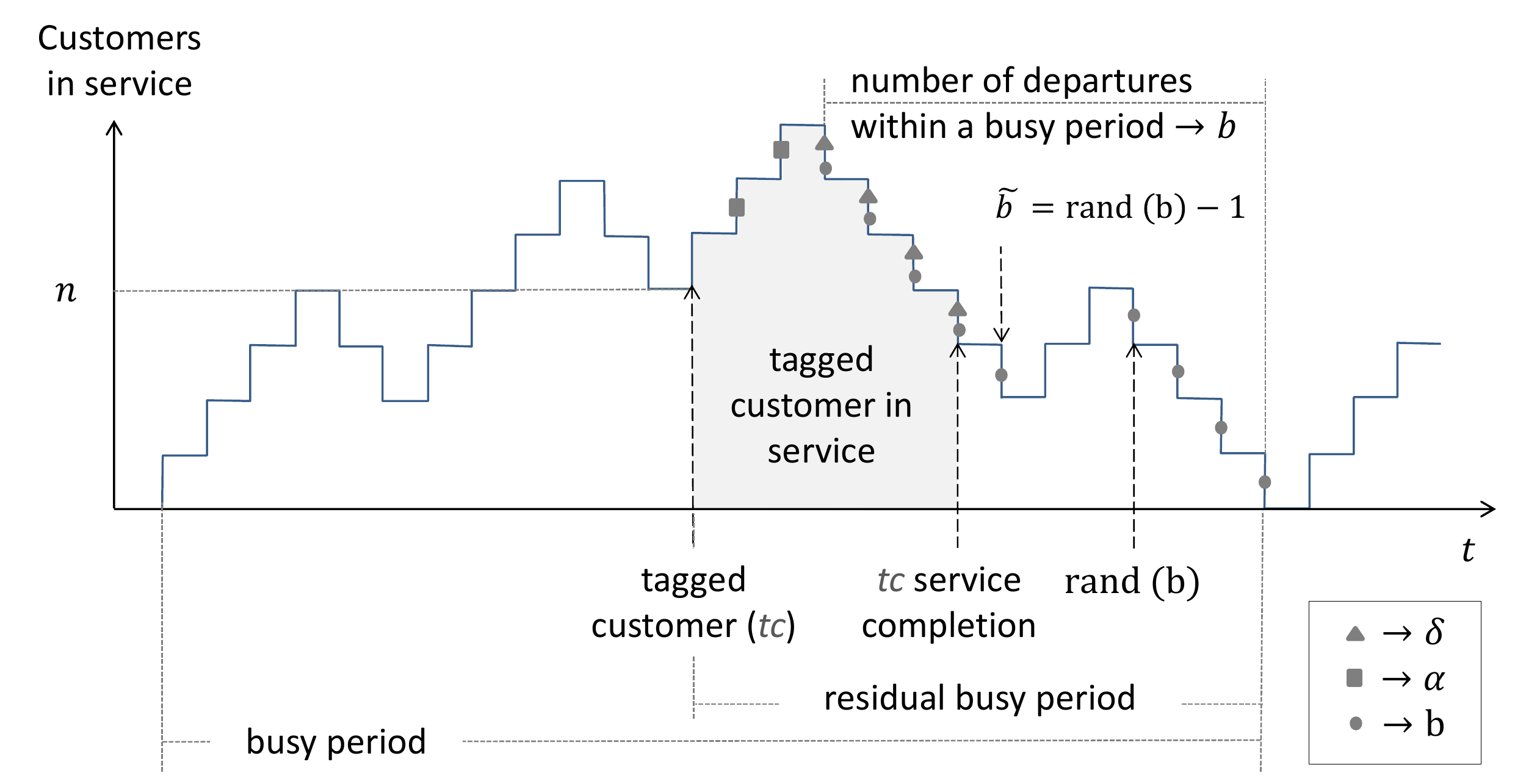}
      \caption{Residual busy period in the queue after the arrival of the tagged customer.}
  \label{fig:busy_period}
\end{figure}

Let $B$ be the number of customers served in the residual busy period and let $\beta$ denote the number of customers served in a regular busy period. It is clear that 
$$
\E(z^B~|~N_0=n) = \beta(z)^{n+1}
$$
since each present customer in the queue at the arrival of the tagged customer generates an independent busy period and then
$$
B(z) = \frac{(1-\rho)\beta(z)}{1-\rho\beta(z)}.
$$

From~\cite[p.~218]{Klein0}, 
$$
\beta(z) = \frac{1+\rho}{2\rho}\left(1-\sqrt{1-\frac{4\rho z}{(1+\rho)^2}}  \right)
$$
and then
$$
B(z) = \frac{(1-\rho)(1+\rho)}{2 \rho^2(z-1)}\left(1-\frac{2\rho z}{1+\rho}-\sqrt{1-\frac{4\rho z}{(1+\rho)^2}}  \right).
$$
The point $z=1$ is obviously a removable singularity for the function $B(z)$ and this function has an algebraic singularity at point $z= \frac{(1+\rho)^2}{4\rho}$. 

By using the expansion
$$
(1-x)^{\frac{1}{2}} = 1-\sum_{k\geq 1} \binom{2k-2}{k-1}\frac{x^k}{2^{2k-1}k},
$$
we obtain
$$
B(z) =  \frac{(1-\rho)(1+\rho)}{ \rho^2 (z-1)}\left(\sum_{k\geq 1} \binom{2k-2}{k-1}\frac{\rho^k z^k}{k(1+\rho)^{2k}}-\frac{\rho z}{1+\rho}  \right) .
$$
Since the term in parenthesis is null for $z=1$, we have
$$
B(z) =  \frac{(1-\rho)(1+\rho)}{ \rho^2}\left(\sum_{k\geq 1} \binom{2k-2}{k-1}\frac{\rho^k }{k(1+\rho)^{2k}}\frac{z^k-z}{z-1}  \right) 
$$
and then
$$
B(z) = \frac{(1-\rho)(1+\rho)}{ \rho^2} \sum_{\ell =0}^\infty z^\ell \sum_{k=\ell+1}^\infty  \binom{2k-2}{k-1}\frac{\rho^k }{k(1+\rho)^{2k}},
$$
where we have used the fact that $\frac{z^k-z}{z-1}= \sum_{\ell = 1}^{k-1} z^\ell$. It follows that the distribution of the number $b$ of customers served in the residual busy period of the tagged customer is defined by: for $\ell \geq 0$
$$
\P(b=\ell) =  \frac{1-\rho}{\rho(1+\rho)} \sum_{k=\ell}^\infty  \binom{2k}{k}\frac{\rho^k }{(k+1)(1+\rho)^{2k}}.
$$

{As shown in Figure~\ref{fig:busy_period},} if we pick up a random customer among those customers served in the residual busy period, then the number of customers served before this customer has distribution $\tilde{b}$ given by
$$
\P(\tilde{b}=j)= \sum_{k=j+1}^\infty \frac{1}{k}\P(b= k )
$$
and generating function
$$
\E\left(z^{\widetilde{b}}\right)=\int_0^1 \frac{B(t)-B(z t )}{t(1-z)} dt.
$$

 For $z$ in the neighborhood of $z_0=\frac{(1+\rho)^2}{4\rho}$, we have
$$
B(z) \sim \frac{1+ \rho}{\rho} -\frac{2(1+\rho)}{\rho(1-\rho)} \sqrt{1-\frac{z}{z_0}}
$$
and by Darboux method~\cite{Flajolet}, we obtain
\begin{equation}
    \label{basymp}
\P({b}=j) \sim \frac{1+\rho}{\rho(1-\rho)\sqrt{\pi}}\frac{1}{j^{\frac{3}{2}}} \left(\frac{2\sqrt{\rho}}{1+\rho}  \right)^{2j}.
\end{equation}
It follows that
\begin{eqnarray*}
\P(\tilde{b}=j)& \sim & \frac{1+\rho}{\rho(1-\rho)\sqrt{\pi}} \sum_{k=j+1}^\infty \frac{1}{k^{\frac{5}{2}}} \left(\frac{1}{z_0}  \right)^k \\
&\sim &  \frac{1+\rho}{\rho(1-\rho)\sqrt{\pi}} \frac{1}{z_0-1} \frac{1}{j^{\frac{5}{2}}} \left(\frac{1}{z_0}  \right)^j  =  \frac{4(1+\rho)}{(1-\rho)^3\sqrt{\pi}}  \frac{1}{j^{\frac{5}{2}}} \left(\frac{1}{z_0}  \right)^j
\end{eqnarray*}

{The probability distribution function of the number of customers served before the one that is randomly picked up in the residual busy period of the tagged customer is shown in Figure~\ref{fig:Pbtilde}. The asymptotic behavior given by Equation~\eqref{basymp}  is illustrated in dashed lines. }

\begin{figure}[hbtp]
	\centering
   \includegraphics[scale=0.45, trim=0 0 0 0, clip] {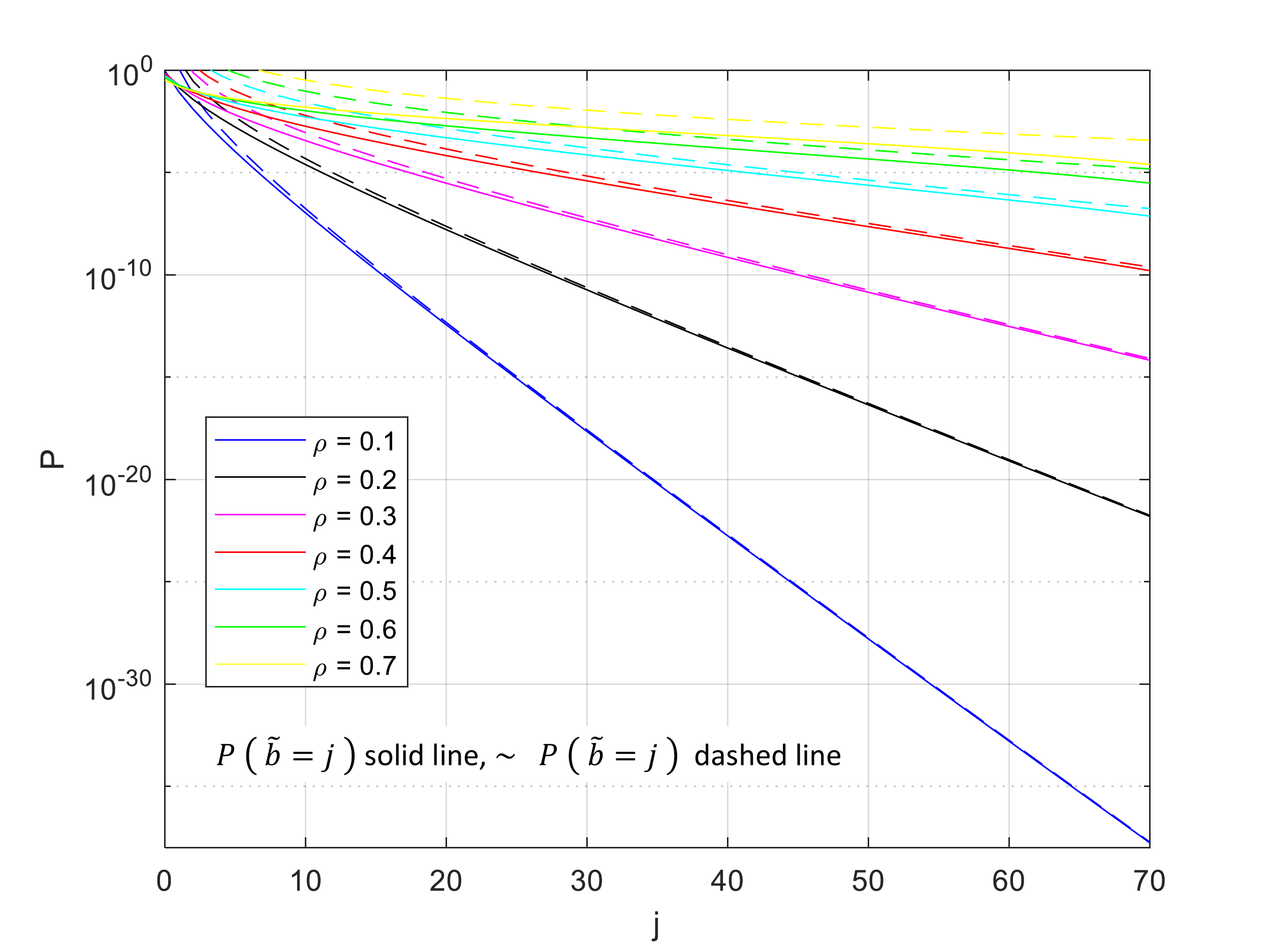}
    \caption{Probability distribution function of $\tilde{b}$. }
  \label{fig:Pbtilde}
\end{figure}

Figure~\ref{fig:Pdelta_btilde} represents the distributions of $\delta$ (the number of departures before service completion of the tagged job) and $\widetilde{b}$ (the number of departures before a customer randomly picked up in the residual busy period).  We can observe that both distribution are reasonably close, especially for light loads. The approximation is less accurate for high loads.

\begin{figure}[hbtp]
	\centering
   \includegraphics[scale=0.45, trim=0 0 0 0, clip] {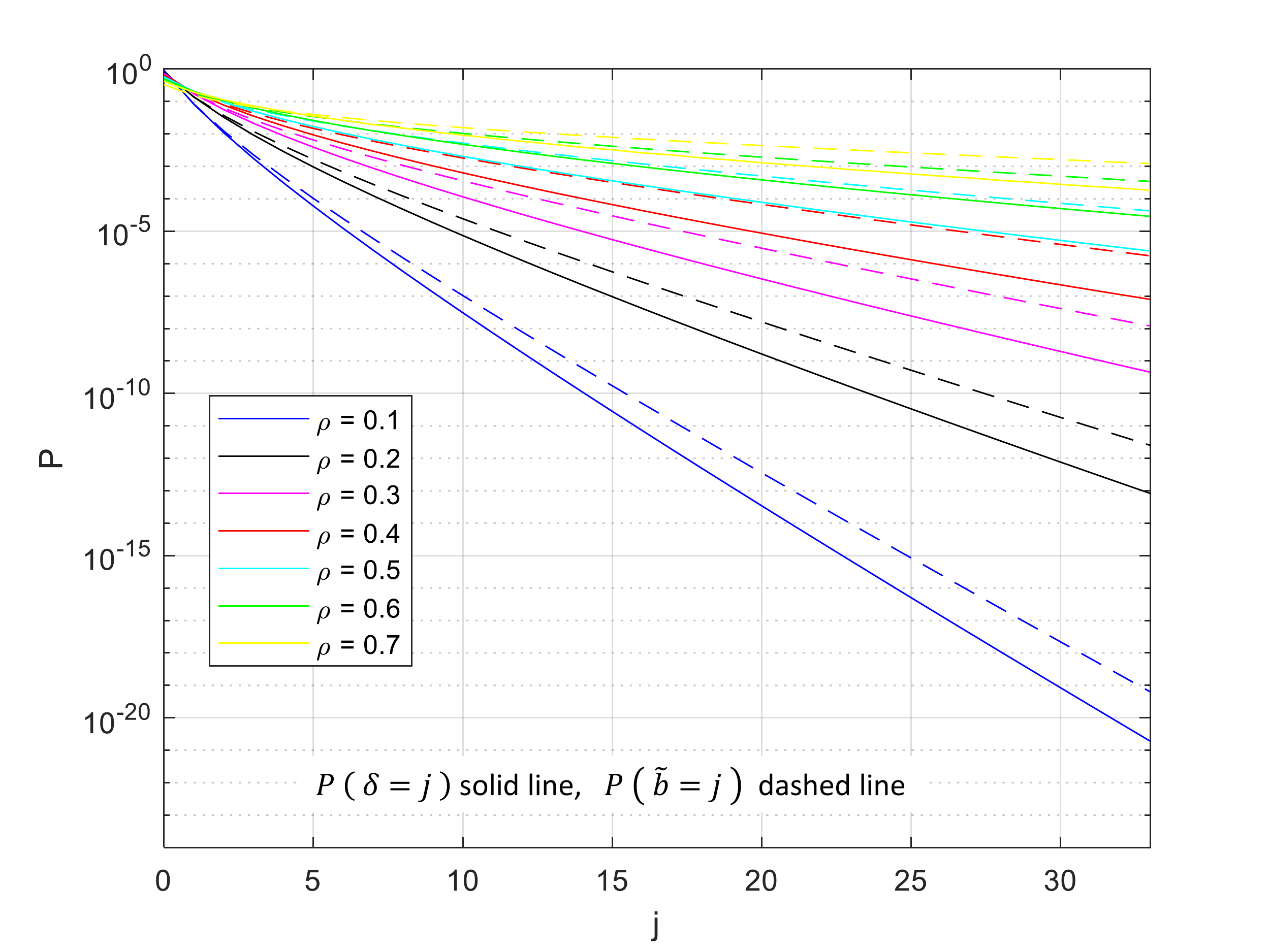}
      \caption{Probability distribution function of $\delta$ and $\tilde{b}$. }
  \label{fig:Pdelta_btilde}
\end{figure}

It is worth noting that $\delta$ and $\widetilde{b}$ have the same decay rate in the sense that
$$
\lim_{j\to \infty}\frac{1}{j}\log \P(\delta=j) = \lim_{j\to \infty}\frac{1}{j}\log \P(\widetilde{b}=j) =\log \frac{4\rho}{(1+\rho)^2}.
$$
When the tagged customer stays for a very long time in the queue, then the residual busy period is also very long and on the same exponential scale as the regular busy period. However, the prefactors are very different and
$$
\lim_{j\to \infty}\frac{\P(\delta=j)}{\P(\widetilde{b}=j)} = 0.
$$
The tagged customers stays much less time in the queue than the time provided by the approximation. In spite of this, we can conclude that the approximation is reasonable for moderate values of $j$ as shown in Figure~\ref{fig:Pdelta_btilde}.

\section{Conclusion}
\label{conclusion}

 We have computed in this paper the number of departures  seen by a tagged customer when entering an $M/M/1$-PS queue in the stationary regime. By using the underlying orthogonal structure of this queue, it is possible to obtain an explicit expression for the distribution of this random variable. We have also derived the asymptotic behavior of this distribution. Subsequently, we have  developed  an approximation for this quantity, which involves random variables with simpler distributions. Numerical experiments show that this approximation is reasonably accurate for moderate values  and asymptotically yields an upper bound for the number of departures seen by the tagged customer in the $M/M/1$-PS queue. The same kind of approximation may be used in more complex systems, such as the $M^{[X]}/M/1$-PS queue.

\bibliographystyle{plain}
\bibliography{adhoc}
\end{document}